\documentclass[12pt,english]{article}
\usepackage{mathptmx}

\usepackage[T1]{fontenc}
\usepackage[a4paper]{geometry}
\usepackage{color}
\usepackage{amsthm}
\usepackage{amsmath}
\usepackage{amssymb}
\usepackage{graphicx}
\usepackage{setspace}
\usepackage[authoryear]{natbib}
\makeatletter

\newcommand{\lyxdot}{.}

\numberwithin{table}{section}
\numberwithin{figure}{section}
\usepackage{enumitem}		
\theoremstyle{plain}
\newtheorem{thm}{\protect\theoremname}[section]
\theoremstyle{definition}
\newtheorem{defn}[thm]{\protect\definitionname}
\theoremstyle{remark}
\newtheorem{rem}[thm]{\protect\remarkname}
\ifx\proof\undefined
\newenvironment{proof}[1][\protect\proofname]{\par
\normalfont\topsep6\p@\@plus6\p@\relax
\trivlist
\itemindent\parindent
\item[\hskip\labelsep
\scshape
#1]\ignorespaces
}{%
\endtrivlist\@endpefalse
}
\providecommand{\proofname}{Proof}
\fi
\theoremstyle{plain}
\newtheorem{prop}[thm]{\protect\propositionname}

\usepackage{amsfonts,amssymb,amsthm,amsbsy,amssymb,amsmath,mathrsfs}
\usepackage{txfonts,graphicx,mathrsfs,rotating,array,sectsty}
\usepackage[symbol]{footmisc}
\usepackage{natbib}
\usepackage{fancyref}
\usepackage[linkcolor=blue,citecolor=blue, colorlinks]{hyperref}
\usepackage{bm}

\newcommand{\q}{\quad}
\newcommand{\qq}{\qquad}
\newcommand{\mcl}[1]{\mathcal{#1}}
\newcommand{\mbd}[1]{\mathbf{#1}}
\newcommand{\mbb}[1]{\mathbb{#1}}
\newcommand{\pr}{\hbox{\sf P}}

\newcommand{\ep}{\hbox{\sf E}}
\newcommand{\var}{\hbox{\sf Var}\,}

\newcommand{\tto}{\rightarrow}
\newcommand{\id}{\hbox{\bf 1}}

\newcommand{\br}[1]{\left(#1\right)}

\newcommand{\dt}{\mathrm{d}}
\newcommand{\ps}{\br{\Omega,\mcl{F},\pr}}

\setcitestyle{round}

\numberwithin{equation}{section}

\allowdisplaybreaks

\sectionfont{\rmfamily\mdseries\large\bf}
\subsectionfont{\rmfamily\mdseries\normalsize\bf}
\subsubsectionfont{\rmfamily\mdseries\normalsize\it}

\@ifundefined{showcaptionsetup}{}{%
 \PassOptionsToPackage{caption=false}{subfig}}
\usepackage{subfig}
\makeatother

\usepackage{babel}
\providecommand{\definitionname}{Definition}
\providecommand{\propositionname}{Proposition}
\providecommand{\remarkname}{Remark}
\providecommand{\theoremname}{Theorem}

\begin{document}

\title{\textbf{\Large Mean-Variance Asset-Liability Management with State-Dependent
Risk Aversion}}

\author{Qian Zhao \thanks{School of Finance and Statistics, East China Normal University, Shanghai, 200241, China; Department of Applied Finance and Actuarial Studies, Faculty of Business and Economics, Macquarie University, Sydney, NSW 2109, Australia. E-mail: qzhao31@gmail.com},
\q Jiaqin Wei\thanks{Corresponding author. Department of Applied Finance and Actuarial Studies, Faculty of Business and Economics, Macquarie University, Sydney, NSW 2109, Australia. E-mail: jiaqinwei@gmail.com},
\q Rongming Wang\thanks{School of Finance and Statistics, and Research Center of International Finance and Risk Management, East China Normal University, Shanghai, 200241, China. E-mail: rmwang@stat.ecnu.edu.cn}}

\date{\normalsize\it\today}
\maketitle
\begin{abstract}
In this paper, we consider the asset-liability management under the
mean-variance criterion. The financial market consists of a risk-free
bond and a stock whose price process is modeled by a geometric Brownian
motion. The liability of the investor is uncontrollable and is modeled
by another geometric Brownian motion. We consider a specific state-dependent
risk aversion which depends on a power function of the liability.
By solving a flow of FBSDEs with bivariate state process, we obtain
the equilibrium strategy among all the open-loop controls for this
time-inconsistent control problem. It shows that the equilibrium strategy
is a feedback control of the liability. 

\end{abstract}

\section{Introduction }

In the pioneer work \citet{m52}, the author considered the portfolio
selection under the well-known mean-variance criterion and derived
the analytical expression of the mean-variance efficient frontier
in the single-period model. This seminal work has become the foundation
of modern portfolio theory and has stimulated numerous extensions. 

On the one hand, some researchers focus on studying the dynamic mean-variance
portfolio selection problem. \citet{s69} considered a discrete-time
multi-period model. More recently, by embedding the original problem
into a stochastic linear-quadratic (LQ) control problem, \citet{ln00}
and \citet{zl00} extended Markowitz's work to a multi-period model
and a continuous-time model, respectively. On the other hand, there
are some works that consider a generalized financial market. An important
and popular subject is the asset and liability management problem,
which studies the selection of portfolio while taking into account
the liabilities of investors. More specifically, in the asset and
liability management, the surplus, i.e. the difference between asset
value and liability value, is considered.

Since it was proposed by \citet{st90} which considered a single-period
model, there is an increasing number of interests in the asset-liability
management under the mean-variance criteria. \citet{km95} studied
the portfolio choice with liabilities and showed that liabilities
affect the efficient frontier. Adopting the embedding technique of
\citet{ln00}, \citet{ltv04} derived an analytical optimal policy
and efficient frontier for the multi-period asset-liability management
problem. The mean-variance asset-liability management in a continuous-time
model was investigated by \citet{cl06} in which a stochastic LQ control
problem was studied and both the optimal strategy and the mean efficient
frontier were obtained. Furthermore, in a regime-switching framework,
\citet{cyy08} and \citet{cy11} studied the mean-variance asset-liability
management in the continuous-time model and mule-period model, respectively.
It is worth to note that, all of these papers suggested that the liabilities
were not controllable, which is the main difference between the Markowitz's
problem and the asset-liability management.

It is well acknowledged that due to the existence of a non-linear
function of the expectation in the objective functional, the mean-variance
portfolio selection problem in a multi-period framework is time inconsistent
in the sense that the Bellman optimality principle does not hold.
Intuitively, an optimal strategy obtained for the initial time may
not be optimal for any latter time. This is the so-called pre-committed
strategy, i.e., the strategy that is only optimal for the initial
time. Note that in all the references we mentioned above (among others),
only the pre-committed strategies have been considered.

In \citet{s55}, the author proposed another approach to study the
time inconsistent problem, i.e., study the problem within a game theoretic
framework by using Nash equilibrium points. Recently, there is an
increasing amount of attention in the time inconsistent control problem
due to the practical applications in the economics and finance. In
\citet{el06} and \citet{ep08} which considered the optimal consumption
and investment problem under hyperbolic discounting, the authors provided
the precise definition of the equilibrium concept in continuous time
for the first time. Following their idea, \citet{bm10} studied the
time-inconsistent control problem in a general Markov framework, and
derived the extended HJB equation together with the verification theorem.
\citet{bmz12} studied the Markowitz's problem with state-dependent
risk aversion by utilizing the extended HJB equation obtained in \citet{bm10}.
They showed that the equilibrium control was dependent on the current
state. Considering a regime-switching model and with the assumption
that the risk aversion depends on the state of the regime, \citet{wwyy12}
investigated the equilibrium strategy for the mean-variance asset-liability
management problem by using the extended HJB equation developed by
\citet{bm10}.

In \citet{el06}, \citet{ep08} and the papers following their idea,
the equilibrium control was defined within the class of feedback controls.
Considering the time-inconsistent stochastic LQ control, \citet{hjz12}
defined the equilibrium control within the class of open-loop controls,
and derived a general sufficient condition for equilibriums through
a flow of forward-backward stochastic differential equations (FBSDEs).
However, the general existence of solutions to the flow of FBSDEs
is an open problem. With the assumption that the state process was
scalar valued and all the coefficients were deterministic, \citet{hjz12}
showed that the flow of FBSDEs could be reduced into several Riccati-like
ordinary differential equations and the equilibrium control could
be obtained explicitly. Also considering the scalar valued state process,
\citet{hjz12} dealt with the Markowitz's problem with state-dependent
risk aversion and stochastic coefficients. Due to the difference between
the definitions of equilibrium controls, their results were rather
different from those obtained in \citet{bm10} and \citet{bmz12}.

Following the idea of \citet{hjz12}, we consider the time-inconsistent
mean-variance asset-liability management. Since the state process
of our problem is bivariate, the solution to the flow of FBSDEs in
\citet{hjz12} can not be directly adopted. We show that the flow
of FBSDEs of our problem can be solved explicitly and the (close-form)
equilibrium strategy can be obtained. There are some differences between
this paper and \citet{wwyy12} which also studied the time-inconsistent
mean-variance asset-liability management. First, the definitions of
equilibrium controls are different. They are inherited from the differences
between \citet{hjz12} and \citet{ep08}. Second, the risk aversion
considered in this paper depends on the liability process (see Remark
\ref{rmk:2-1}), while the risk aversion in \citet{wwyy12} only depends
on the state of regime and it becomes constant when there is only
one regime. Since the risk aversion is independent of the surplus
process, the equilibrium strategy in this paper is a feedback control
of the liability process which is similar to \citet{wwyy12}. Although
we use different definitions of the equilibrium strategy from \citet{wwyy12},
in a special case we get the same result with \citet{wwyy12} (see
Remark \ref{rmk:3-2}).

The remainder of this paper is organized as follows. Section 2 introduces
the model, the definition of the equilibrium strategy and the flow
of FBSDEs of our problem. In section 3 we derive the solution to the
flow of FBSDEs and the equilibrium strategy. Section 4 establishes
the equilibrium value function. Some numerical examples are illustrated
in section 5.

\section{Preliminaries}

\subsection{The model}

Let $\ps$ be a fixed complete probability space on which two independent
standard Brownian motions $W_{1}(t)$ and $W_{2}(t)$ are defined.
Let $T>0$ be the fixed and finite time horizon and denote by $\left\{ \mcl F_{t}\right\} _{t\in[0,T]}$
the augmented filtration generated by $(W_{1}(t),W_{2}(t))$. 

We introduce the following notation with $n$ being a generic integer:
\begin{align*}
L_{\mcl G}^{2}(\Omega;\mbb R^{n}): & \q\text{the set of random variables }\xi:(\Omega,\mcl G)\tto(\mbb R^{n},\mcl B(\mbb R^{n}))\text{ with \ep\ensuremath{\left[\left|\xi\right|^{2}\right]}<+\ensuremath{\infty}.}\\
L_{\mcl G}^{2}(t,T;\mbb R^{n}): & \q\text{the set of }\{\mcl G\}_{s\in[t,T]}\text{-adapted processes }\{f(s)\}_{s\in[t,T]}\\
 & \q\text{ with }\ep\left[\int_{t}^{T}\left|f(s)\right|^{2}\dt s\right]<\infty.\\
L_{\mcl G}^{2}(\Omega;C(t,T,\mbb R^{n})): & \q\text{the set of continuous }\{\mcl G\}_{s\in[t,T]}\text{-adapted processes }\{f(s)\}_{s\in[t,T]}\\
 & \q\text{ with }\ep\left[\sup_{s\in[t,T]}\left|f(s)\right|^{2}\right]<\infty.
\end{align*}

In what follows, unless otherwise specified, we adopt bold-face letters
to denote matrices and vectors, and the transpose of a matrix or vector
$\mbd M$ is denoted by $\mbd M'$. Also, we denote by $M_{ij}$ (or
$M_{i}$) the $(i,j)$-element (or the $i$-th element) of the matrix
$\mbd M$ (or the vector $\mbd M$). 

We consider a financial market consisting of one bond and one stock
within the time horizon $[0,T]$. The price of the risk-free bond
$B(t)$ satisfies 
\[
\dt B(t)=r(t)B(t)\dt t,\q B(0)=1,\q0\leq t\leq T.
\]
The price of the stock $P(t)$ is given by 
\[
\dt P(t)=P(t)\left[\mu(t)\dt t+\sigma(t)\dt W_{1}(t)\right],\q0\leq t\leq T,
\]
where $P(0)=p_{0}>0.$ 

Denote by $L(t)$ the liability of the investor. We assume that the
liability and the stock price are correlated and the dynamics of liability
is given by 
\[
\dt L(t)=L(t)\left[\alpha(t)\dt t+\rho(t)\beta(t)\dt W_{1}(t)+\sqrt{1-\rho^{2}(t)}\beta(t)\dt W_{2}(t)\right],\q0\leq t\leq T,
\]
where $L(0)=l_{0}>0$ and $\rho(t)\in[0,1]$ for all $t\in[0,T]$.

Let $u(t)$ be the dollar amount invested in the stock at time $t$.
Then the asset in the stock market $Z(t)$ evolves as
\[
\dt Z(t)=\left[r(t)Z(t)+(\mu(t)-r(t))u(t)\right]\dt t+\sigma(t)u(t)\dt W_{1}(t),\q0\leq t\leq T,
\]
where $Z(0)=z_{0}.$ The surplus process for the asset-liability management
is given by $S(t):=Z(t)-L(t).$ Then the dynamics of $S(t)$ is 
\begin{eqnarray*}
\dt S(t) & = & \left[r(t)S(t)+\eta(t)L(t)+\theta(t)u(t)\right]\dt t+\left[\sigma(t)u(t)-\rho(t)\beta(t)L(t)\right]\dt W_{1}(t)\\
 &  & -\sqrt{1-\rho^{2}(t)}\beta(t)L(t)\dt W_{2}(t),\qq0\leq t\leq T,
\end{eqnarray*}
where $\eta(t)=r(t)-\alpha(t)$, $\theta(t)=\mu(t)-r(t)$ and $S(0)=z_{0}-l_{0}:=s_{0}$. 

Let $\mbd X(t)=(S(t),L(t))'$ be the bivariate state process and $\mbd X(0)=\mbd x_{0}:=(s_{0},l_{0})'$.
Thus we have
\begin{equation}
\dt\mbd X(t)=\left[\mbd A(t)\mbd X(t)+\mbd B'(t)u(t)\right]\dt t+\left[\mbd C_{1}(t)\mbd X(t)+\mbd D(t)u(t)\right]\dt W_{1}(t)+\mbd C_{2}(t)\mbd X(t)\dt W_{2}(t),\label{eq:2-1-0}
\end{equation}
where
\[
\mbd A(t)=\begin{pmatrix}r(t) & \eta(t)\\
0 & \alpha(t)
\end{pmatrix},\q\mbd B(t)=(\theta(t),0),\q\mbd C_{1}(t)=\begin{pmatrix}0 & -\rho(t)\beta(t)\\
0 & \rho(t)\beta(t)
\end{pmatrix},\q\mbd C_{2}(t)=\begin{pmatrix}0 & -\sqrt{1-\rho^{2}(t)}\beta(t)\\
0 & \sqrt{1-\rho^{2}(t)}\beta(t)
\end{pmatrix},
\]
and $\mbd D(t)=(\sigma(t),0)'$. We assume that $\mbd A$ $\mbd B,\mbd C_{1},\mbd C_{2}$
and $\mbd D$ are bounded deterministic functions on $[0,T]$ valued
in $\mbb R^{2\times2}$, $\mbb R^{1\times2},\mbb R^{2\times2},\mbb R^{2\times2}$
and $\mbb R^{2\times1}$, respectively.
\begin{defn}
A strategy $u$ is said to be admissible if $u\in L_{\mcl F}^{2}(0,T;\mbb R)$
such that SDE (\ref{eq:2-1-0}) has a unique solution $\mbd X\in L_{\mcl F}^{2}(\Omega;C(0,T,\mbb R^{2}))$.
\end{defn}
For the time-inconsistent control problem, we will consider the controlled
state process starting from time $t\in[0,T]$ and state $\mbd x_{t}\in L_{\mcl F}^{2}(\Omega,\mbb R^{2})$:
\begin{equation}
\dt\mbd X(s)=\left[\mbd A(s)\mbd X(s)+\mbd B'(s)u(s)\right]\dt s+\left[\mbd C_{1}(s)\mbd X(s)+\mbd D(s)u(s)\right]\dt W_{1}(s)+\mbd C_{2}(s)\mbd X(s)\dt W_{2}(s),\label{eq:2-1-1}
\end{equation}
with $\mbd X(t)=\mbd x_{t}.$ Note that for any strategy $u\in L_{\mcl F}^{2}(t,T;\mbb R)$,
SDE (\ref{eq:2-1-1}) admits a unique solution $\mbd X^{t,\mbd x_{t},u}\in L_{\mcl F}^{2}(\Omega;C(t,T,\mbb R^{2}))$.

At any initial state $(t,\mbd x_{t})$, the mean-variance cost functional
is given by
\begin{eqnarray}
J(t,\mbd x_{t};u) & := & \frac{1}{2}\var_{t}\left[S(T)\right]-\left[\omega_{1}L^{-\lambda}(t)+\omega_{2}\right]\ep_{t}\left[S(T)\right]\nonumber \\
 & = & \frac{1}{2}\ep_{t}\left[S^{2}(T)\right]-\frac{1}{2}\left(\ep_{t}\left[S(T)\right]\right)^{2}-\left[\omega_{1}L^{-\lambda}(t)+\omega_{2}\right]\ep_{t}\left[S(T)\right],\label{eq:2-1-2}
\end{eqnarray}
 where $u\in L_{\mcl F}^{2}(t,T;\mbb R)$, $(S,L)'=\mbd X^{t,\mbd x_{t},u}$,
$\omega_{1}$ , $\omega_{2}$ , $\lambda$ are nonnegative constants,
and $\ep_{t}[\cdot]:=\ep[\cdot\mid\mcl F_{t}]$. 
\begin{rem}
\textcolor{black}{\label{rmk:2-1}Note that $\frac{1}{\omega_{1}L^{-\lambda}(t)+\omega_{2}}$
is a state-dependent risk aversion of the investor. Taking $\omega_{1}\geq0$
and $\lambda\geq0$ implies that the risk aversion increases with
increasing liability which is reasonable for a common investor. Noting
that, with such a risk aversion, the investor is uniformly risk averse.}
\end{rem}

\subsection{The equilibrium strategy}

In this subsection, we introduce the equilibrium strategy to the time-inconsistent
control problem. We use the definition of the equilibrium strategy
from \citet{hjz12}.
\begin{defn}
Let $u^{*}\in L_{\mcl F}^{2}(0,T;\mbb R)$ be a given strategy and
$\mbd X^{*}$ be the state process corresponding to $u^{*}$. The
strategy $u^{*}$ is called an equilibrium strategy if for any $t\in[0,T)$
and $v\in L_{\mcl F_{t}}^{2}(\Omega,\mbb R)$, 
\[
\liminf_{\epsilon\tto0}\frac{J(t,\mbd X^{*}(t);u^{t,\epsilon,v})-J(t,\mbd X^{*}(t);u^{*})}{\epsilon}\geq0,
\]
where
\[
u^{t,\epsilon,v}(s):=u^{*}(s)+v\id_{s\in[t,t+\epsilon]},\q s\in[t,T],
\]
for any $t\in[0,T)$ and $\epsilon>0$. The equilibrium value function
is defined by
\begin{equation}
V(t,\mbd X^{*}(t)):=J(t,\mbd X^{*}(t);u^{*}).\label{eq:2-2-0-1}
\end{equation}

\end{defn}
Although we have stated the difference between definitions of equilibrium
strategy in \citet{hjz12} and \citet{el06} in previous section,
they have similar intuition. We refer the reader to these papers for
more details.

Let $u^{*}$ be a fixed strategy and $\mbd X^{*}$ be the corresponding
state process. For any $t\in[0,T),$ the adjoint process $(\mbd p(\cdot;t),(\mbd k_{1}(\cdot;t),\mbd k_{2}(\cdot;t)))\in L_{\mcl F}^{2}(t,T;\mbb R^{2})\times\left(L_{\mcl F}^{2}(t,T;\mbb R^{2})\right)^{2}$
is defined in the time interval $[t,T]$ by 
\begin{equation}
\begin{cases}
\dt\mbd p(s;t) & =-\left[\mbd A'(s)\mbd p(s;t)+\mbd C_{1}'(s)\mbd k_{1}(s;t)+\mbd C_{2}'(s)\mbd k_{2}(s;t)\right]\dt s+\sum_{i=1}^{2}\mbd k_{i}(s;t)\dt W_{i}(s),\q s\in[t,T],\\
\mbd p(T;t) & =\mbd G\mbd X^{*}(T)-\mbd h\ep_{t}\left[\mbd X^{*}(T)\right]-\left[\omega_{1}L^{-\lambda}(t)+\omega_{2}\right]\mbd e,
\end{cases}\label{eq:2-2-1}
\end{equation}
where
\[
\mbd G=\mbd h=\begin{pmatrix}1 & 0\\
0 & 0
\end{pmatrix},\q\mbd e=\begin{pmatrix}1\\
0
\end{pmatrix}.
\]

Note that the risk aversion in our model is different from \citet{hjz12}
in which the reciprocal of the risk aversion is a linear function
of the state process. However, with $\mbd p(s;t)$ defined by (\ref{eq:2-2-1}),
Proposition 3.1 in \citet{hjz12} still holds for our model. Hence,
we have the following result which gives a sufficient condition of
equilibrium strategies for our asset-liability management problem.
\begin{thm}
\label{thm:2-1}A strategy $u^{*}\in L_{\mcl F}^{2}(0,T;\mbb R)$
is an equilibrium strategy if for any time $t\in[0,T]$:
\begin{enumerate}[label=(\roman{enumi})]
\item  the system of stochastic differential equations
\begin{equation}
\begin{cases}
\dt\mbd X^{*}(s) & =\left[\mbd A(s)\mbd X^{*}(s)+\mbd B'(s)u^{*}(s)\right]\dt s+\left[\mbd C_{1}(s)\mbd X^{*}(s)+\mbd D(s)u^{*}(s)\right]\dt W_{1}(s)+\mbd C_{2}(s)\mbd X^{*}(s)\dt W_{2}(s),\\
\mbd X^{*}(0) & =\left(s_{0},l_{0}\right)';\\
\dt\mbd p(s;t) & =-\left[\mbd A'(s)\mbd p(s;t)+\mbd C_{1}'(s)\mbd k_{1}(s;t)+\mbd C_{2}'(s)\mbd k_{2}(s;t)\right]\dt s+\sum_{i=1}^{2}\mbd k_{i}(s;t)\dt W_{i}(s),\q s\in[t,T],\\
\mbd p(T;t) & =\mbd G\mbd X^{*}(T)-\mbd h\ep_{t}\left[\mbd X^{*}(T)\right]-\left[\omega_{1}L^{-\lambda}(t)+\omega_{2}\right]\mbd e;
\end{cases}\label{eq:2-3-4}
\end{equation}
admits a solution $\left(\mbd X^{*},\mbd p,(\mbd k_{1},\mbd k_{2})\right)$;
\item $\Lambda(s;t)=\mbd B(s)\mbd p(s;t)+\mbd D'(s)\mbd k_{1}(s;t)$ satisfies
\begin{equation}
\ep_{t}\left[\int_{t}^{T}\left|\Lambda(s;t)\right|\dt s\right]<\infty,\q\lim_{s\downarrow t}\ep_{t}\left[\Lambda(s;t)\right]=0,\; a.s.,\;\forall t\in[0,T].\label{eq:2-3-5}
\end{equation}

\end{enumerate}
\end{thm}
As mentioned by \citet{hjz12}, under some condition, the second equality
in (\ref{eq:2-3-5}) is ensured by 
\begin{equation}
\mbd B(t)\mbd p(t;t)+\mbd D'(t)\mbd k_{1}(t,t)=0.\label{eq:2-3-6}
\end{equation}

From the above theorem, if we can solve the flow of FBSDEs (\ref{eq:2-3-4}),
then we can get the equilibrium strategy. However, the general result
for the solution to a flow of FBSDEs is not available. In the next
section, we will solve the flow of FBSDEs (\ref{eq:2-3-4}) with bivariate
state process.

\section{The Equilibrium Strategy}

\subsection{The solution to the flow of FBSDEs (\ref{eq:2-3-4})}

Let $\mbd p=(p_{1},p_{2})',\mbd k_{1}=(k_{1,1},k_{1,2})'$ and $\mbd k_{2}=(k_{2,1},k_{2,2})'$.
We rewrite (\ref{eq:2-3-4}) and (\ref{eq:2-3-6}) as
\begin{equation}
\begin{cases}
\dt S^{*}(s) & =\left[r(s)S^{*}(s)+\eta(s)L(s)+\theta(s)u^{*}(s)\right]\dt s+\left[\sigma(s)u^{*}(s)-\rho(s)\beta(s)L(s)\right]\dt W_{1}(s)\\
 & \q-\sqrt{1-\rho^{2}(s)}\beta(s)L(s)\dt W_{2}(s),\qq0\leq s\leq T,\\
S^{*}(0) & =s_{0},\\
\dt L(s) & =L(s)\left[\alpha(s)\dt s+\rho(s)\beta(s)\dt W_{1}(s)+\sqrt{1-\rho^{2}(s)}\beta(s)\dt W_{2}(s)\right],\q0\leq s\leq T,\\
L(0) & =l_{0},\\
\dt p_{1}(s;t) & =-r(s)p_{1}(s;t)\dt s+k_{1,1}(s;t)\dt W_{1}(s)+k_{2,1}(s;t)\dt W_{2}(s),\q s\in[t,T],\\
p_{1}(T;t) & =S^{*}(T)-\ep_{t}\left[S^{*}(T)\right]-\left[\omega_{1}L^{-\lambda}(t)+\omega_{2}\right],\\
\dt p_{2}(s;t) & =-\left\{ \eta(s)p_{1}(s;t)+\alpha(s)p_{2}(s;t)-\rho(s)\beta(s)\left[k_{1,1}(s;t)-k_{1,2}(s;t)\right]\right.\\
 & \q\left.-\sqrt{1-\rho^{2}(s)}\beta(s)\left[k_{2,1}(s;t)-k_{2,2}(s;t)\right]\right\} \\
 & \q+k_{1,2}(s;t)\dt W_{1}(s)+k_{2,2}(s;t)\dt W_{2}(s),\q s\in[t,T],\\
p_{2}(T;t) & =0
\end{cases}\label{eq:2-4-1}
\end{equation}
and 
\begin{equation}
\theta(t)p_{1}(t;t)+\sigma(t)k_{1,1}(t;t)=0,\label{eq:2-4-0}
\end{equation}
respectively.

Similar to \citet{hjz12}, we consider the following ansatz:
\begin{eqnarray}
p_{1}(s;t) & = & M_{1}(s)L^{-\lambda}(s)+M_{2}(s)S^{*}(s)+M_{3}(s)L(s)\nonumber \\
 &  & +M_{4}(s)\ep_{t}\left[L^{-\lambda}(s)\right]+M_{5}(s)\ep_{t}\left[S^{*}(s)\right]+M_{6}(s)\ep_{t}\left[L(s)\right]\nonumber \\
 &  & +M_{7}(s)L^{-\lambda}(t)+M_{8}(s)S^{*}(t)+M_{9}(s)L(t)+M_{10}(s),\label{eq:2-4-2}\\
p_{2}(s;t) & = & N_{1}(s)L^{-\lambda}(s)+N_{2}(s)S^{*}(s)+N_{3}(s)L(s)\nonumber \\
 &  & +N_{4}(s)\ep_{t}\left[L^{-\lambda}(s)\right]+N_{5}(s)\ep_{t}\left[S^{*}(s)\right]+N_{6}(s)\ep_{t}\left[L(s)\right]\nonumber \\
 &  & +N_{7}(s)L^{-\lambda}(t)+N_{8}(s)S^{*}(t)+N_{9}(s)L(t)+N_{10}(s),\label{eq:2-4-3}
\end{eqnarray}
where $M_{i}$ and $N_{i}$, $i=1,\cdots,10$, are deterministic differentiable
functions with $\dot{M}_{i}=m_{i}$ and $\dot{N}_{i}=n_{i}$, $i=1,\cdots,10$.
In the following, we get the solutions to $ $$M_{i},i=1,\cdots,10$.
The derivation for $N_{i},i=1,\cdots,10$ are similar, and since they
will not appear in the equilibrium strategy or the equilibrium value
function, we omit the details.

By Itô's formula, it is easy to see that
\begin{eqnarray*}
\dt L^{-\lambda}(s) & = & -\lambda L^{-(\lambda+1)}(s)\dt L(s)+\frac{1}{2}\lambda(\lambda+1)L^{-(\lambda+2)}(s)\dt[L,L](s)\\
 & = & -\lambda L^{-\lambda}(s)\left\{ \left[\alpha(s)-\frac{1}{2}(\lambda+1)\beta^{2}(s)\right]\dt s+\left[\rho(s)\beta(s)\dt W_{1}(s)+\sqrt{1-\rho^{2}(s)}\beta(s)\dt W_{2}(s)\right]\right\} .
\end{eqnarray*}
Consequently, we have
\begin{eqnarray}
\dt p_{1}(s;t) & = & m_{1}(s)L^{-\lambda}(s)\dt s-\lambda M_{1}(s)L^{-\lambda}(s)\left[\alpha(s)-\frac{1}{2}(\lambda+1)\beta^{2}(s)\right]\dt s\nonumber \\
 &  & -\lambda M_{1}(s)L^{-\lambda}(s)\left[\rho(s)\beta(s)\dt W_{1}(s)+\sqrt{1-\rho^{2}(s)}\beta(s)\dt W_{2}(s)\right]\nonumber \\
 &  & +m_{2}(s)S^{*}(s)\dt s+M_{2}(s)\left[r(s)S^{*}(s)+\eta(s)L(s)+\theta(s)u^{*}(s)\right]\dt s\nonumber \\
 &  & +M_{2}(s)\left\{ \left[\sigma(s)u^{*}(s)-\rho(s)\beta(s)L(s)\right]\dt W_{1}(s)-\sqrt{1-\rho^{2}(s)}\beta(s)L(s)\dt W_{2}(s)\right\} \nonumber \\
 &  & +m_{3}(s)L(s)\dt s+M_{3}(s)L(s)\left[\alpha(s)\dt s+\rho(s)\beta(s)\dt W_{1}(s)+\sqrt{1-\rho^{2}(s)}\beta(s)\dt W_{2}(s)\right]\nonumber \\
 &  & +m_{4}(s)\ep_{t}\left[L^{-\lambda}(s)\right]\dt s-\lambda\left[\alpha(s)-\frac{1}{2}(\lambda+1)\beta^{2}(s)\right]M_{4}(s)\ep_{t}\left[L^{-\lambda}(s)\right]\dt s\nonumber \\
 &  & +m_{5}(s)\ep_{t}\left[S^{*}(s)\right]\dt s+M_{5}(s)\ep_{t}\left[r(s)S^{*}(s)+\eta(s)L(s)+\theta(s)u^{*}(s)\right]\dt s\nonumber \\
 &  & +m_{6}(s)\ep_{t}\left[L(s)\right]\dt s+\alpha(s)M_{6}(s)\ep_{t}\left[L(s)\right]\dt s\nonumber \\
 &  & +m_{7}(s)L^{-\lambda}(t)\dt s+m_{8}(s)S^{*}(t)\dt s+m_{9}(s)L(t)\dt s+m_{10}(s)\dt s.\label{eq:2-4-5}
\end{eqnarray}
Comparing the $\dt W_{1}(s)$-term and $\dt W_{2}(s)$-term in (\ref{eq:2-4-5})
and (\ref{eq:2-4-1}), we obtain
\begin{eqnarray}
\begin{cases}
k_{1,1}(s;t) & =-\lambda\rho(s)\beta(s)M_{1}(s)L^{-\lambda}(s)+M_{2}(s)\left[\sigma(s)u^{*}(s)-\rho(s)\beta(s)L(s)\right]\\
 & \q+\rho(s)\beta(s)M_{3}(s)L(s),\\
k_{2,1}(s;t) & =-\lambda\sqrt{1-\rho^{2}(s)}\beta(s)M_{1}(s)L^{-\lambda}(s)-\sqrt{1-\rho^{2}(s)}\beta(s)M_{2}(s)L(s)\\
 & \q+\sqrt{1-\rho^{2}(s)}\beta(s)M_{3}(s)L(s).
\end{cases}\label{eq:2-4-6}
\end{eqnarray}
Putting $p_{1}$ and $k_{11}$ into (\ref{eq:2-4-0}), it yields that
\begin{align*}
\theta(s)\left\{ M_{1}(s)L^{-\lambda}(s)+M_{2}(s)S^{*}(s)+M_{3}(s)L(s)\right.\\
+M_{4}(s)L^{-\lambda}(s)+M_{5}(s)S^{*}(s)+M_{6}(s)L(s)\\
\left.+M_{7}(s)L^{-\lambda}(s)+M_{8}(s)S^{*}(s)+M_{9}(s)L(s)+M_{10}(s)\right\} \\
+\sigma(s)\left\{ -\rho(s)\beta(s)\left[\lambda M_{1}(s)L^{-\lambda}(s)+\left(M_{2}(s)-M_{3}(s)\right)L(s)\right]+\sigma(s)M_{2}(s)u^{*}(s)\right\}  & =0,
\end{align*}
i.e.
\begin{align*}
\left\{ \theta(s)\left[M_{1}(s)+M_{4}(s)+M_{7}(s)\right]-\lambda\sigma(s)\rho(s)\beta(s)M_{1}(s)\right\} L^{-\lambda}(s)\\
+\theta(s)\left[M_{2}(s)+M_{5}(s)+M_{8}(s)\right]S^{*}(s)\\
\left\{ \theta(s)\left[M_{3}(s)+M_{6}(s)+M_{9}(s)\right]-\sigma(s)\rho(s)\beta(s)\left[M_{2}(s)-M_{3}(s)\right]\right\} L(s)\\
+\theta(s)M_{10}(s)+\sigma^{2}(s)M_{2}(s)u^{*}(s) & =0,
\end{align*}
which implies 
\begin{align*}
u^{*}(s) & =f_{1}(s)L^{-\lambda}(s)+f_{2}(s)S^{*}(s)+f_{3}(s)L(s)+f_{4}(s),\q0\leq s\leq T,
\end{align*}
 where
\begin{align}
\begin{cases}
f_{1}(s) & =-\frac{\theta(s)\left[M_{1}(s)+M_{4}(s)+M_{7}(s)\right]-\lambda\sigma(s)\rho(s)\beta(s)M_{1}(s)}{\sigma^{2}(s)M_{2}(s)},\\
f_{2}(s) & =-\frac{\theta(s)\left[M_{2}(s)+M_{5}(s)+M_{8}(s)\right]}{\sigma^{2}(s)M_{2}(s)},\\
f_{3}(s) & =-\frac{\theta(s)\left[M_{3}(s)+M_{6}(s)+M_{9}(s)\right]-\sigma(s)\rho(s)\beta(s)\left[M_{2}(s)-M_{3}(s)\right]}{\sigma^{2}(s)M_{2}(s)},\\
f_{4}(s) & =-\frac{\theta(s)M_{10}(s)}{\sigma^{2}(s)M_{2}(s)}.
\end{cases}\label{eq:2-4-7}
\end{align}
Comparing the $\dt s$-term of $p_{1}(s;t)$ in (\ref{eq:2-4-1})
and (\ref{eq:2-4-5}), we get 
\begin{align*}
r(s)\left\{ M_{1}(s)L^{-\lambda}(s)+M_{2}(s)S^{*}(s)+M_{3}(s)L(s)\right.\\
+M_{4}(s)\ep_{t}\left[L^{-\lambda}(s)\right]+M_{5}(s)\ep_{t}\left[S^{*}(s)\right]+M_{6}(s)\ep_{t}\left[L(s)\right]\\
\left.+M_{7}(s)L^{-\lambda}(t)+M_{8}(s)S^{*}(t)+M_{9}(s)L(t)+M_{10}(s)\right\} \\
+m_{1}(s)L^{-\lambda}(s)-\lambda\left[\alpha(s)-\frac{1}{2}(\lambda+1)\beta^{2}(s)\right]M_{1}(s)L^{-\lambda}(s)\\
+m_{2}(s)S^{*}(s)+M_{2}(s)\left[r(s)S^{*}(s)+\eta(s)L(s)+\theta(s)u^{*}(s)\right]\\
+m_{3}(s)L(s)+\alpha(s)M_{3}(s)L(s)\\
+m_{4}(s)\ep_{t}\left[L^{-\lambda}(s)\right]-\lambda\left[\alpha(s)-\frac{1}{2}(\lambda+1)\beta^{2}(s)\right]M_{4}(s)\ep_{t}\left[L^{-\lambda}(s)\right]\\
+m_{5}(s)\ep_{t}\left[S^{*}(s)\right]+M_{5}(s)\ep_{t}\left[r(s)S^{*}(s)+\eta(s)L(s)+\theta(s)u^{*}(s)\right]\\
+m_{6}(s)\ep_{t}\left[L(s)\right]+\alpha(s)M_{6}(s)\ep_{t}\left[L(s)\right]\\
+m_{7}(s)L^{-\lambda}(t)+m_{8}(s)S^{*}(t)+m_{9}(s)L(t)+m_{10}(s) & =0,
\end{align*}
 i.e.,
\begin{align*}
\left\{ m_{1}(s)+\left(r(s)-\lambda\left[\alpha(s)-\frac{1}{2}(\lambda+1)\beta^{2}(s)\right]\right)M_{1}(s)\right\} L^{-\lambda}(s)\\
+\left\{ m_{2}(s)+2r(s)M_{2}(s)\right\} S^{*}(s)\\
+\left\{ m_{3}(s)+\left[r(s)+\alpha(s)\right]M_{3}(s)+\eta(s)M_{2}(s)\right\} L(s)\\
+\left\{ m_{4}(s)+\left(r(s)-\lambda\left[\alpha(s)-\frac{1}{2}(\lambda+1)\beta^{2}(s)\right]\right)M_{4}(s)\right\} \ep_{t}\left[L^{-\lambda}(s)\right]\\
\left\{ m_{5}(s)+2r(s)M_{5}(s)\right\} \ep_{t}\left[S^{*}(s)\right]\\
\left\{ m_{6}(s)+\left[r(s)+\alpha(s)\right]M_{6}(s)+\eta(s)M_{5}(s)\right\} \ep_{t}\left[L(s)\right]\\
+\left\{ m_{7}(s)+r(s)M_{7}(s)\right\} L^{-\lambda}(t)+\left\{ m_{8}(s)+r(s)M_{8}(s)\right\} S^{*}(t)\\
+\left\{ m_{9}(s)+r(s)M_{9}(s)\right\} L(t)+m_{10}(s)+r(s)M_{10}(s)\\
+\theta(s)M_{2}(s)\left[f_{1}(s)L^{-\lambda}(s)+f_{2}(s)S^{*}(s)+f_{3}(s)L(s)+f_{4}(s)\right]\\
+\theta(s)M_{5}(s)\ep_{t}\left[f_{1}(s)L^{-\lambda}(s)+f_{2}(s)S^{*}(s)+f_{3}(s)L(s)+f_{4}(s)\right] & =0.
\end{align*}
Putting (\ref{eq:2-4-7}) into the above equation, we have
\begin{align}
\left\{ m_{1}(s)+\left(r(s)-\lambda\left[\alpha(s)-\frac{1}{2}(\lambda+1)\beta^{2}(s)\right]\right)M_{1}(s)\right.\nonumber \\
\left.-\frac{\theta^{2}(s)\left[M_{1}(s)+M_{4}(s)+M_{7}(s)\right]-\lambda\theta(s)\sigma(s)\rho(s)\beta(s)M_{1}(s)}{\sigma^{2}(s)}\right\} L^{-\lambda}(s)\nonumber \\
+\left\{ m_{2}(s)+2r(s)M_{2}(s)-\frac{\theta^{2}(s)\left[M_{2}(s)+M_{5}(s)+M_{8}(s)\right]}{\sigma^{2}(s)}\right\} S^{*}(s)\nonumber \\
+\left\{ m_{3}(s)+\left[r(s)+\alpha(s)\right]M_{3}(s)+\eta(s)M_{2}(s)\right.\nonumber \\
\left.-\frac{\theta^{2}(s)\left[M_{3}(s)+M_{6}(s)+M_{9}(s)\right]-\theta(s)\sigma(s)\rho(s)\beta(s)\left[M_{2}(s)-M_{3}(s)\right]}{\sigma^{2}(s)}\right\} L(s)\nonumber \\
+\left\{ m_{4}(s)+\left(r(s)-\lambda\left[\alpha(s)-\frac{1}{2}(\lambda+1)\beta^{2}(s)\right]\right)M_{4}(s)\right.\nonumber \\
\left.-M_{5}(s)\frac{\theta^{2}(s)\left[M_{1}(s)+M_{4}(s)+M_{7}(s)\right]-\lambda\theta(s)\sigma(s)\rho(s)\beta(s)M_{1}(s)}{M_{2}(s)\sigma^{2}(s)}\right\} \ep_{t}\left[L^{-\lambda}(s)\right]\nonumber \\
+\left\{ m_{5}(s)+2r(s)M_{5}(s)-M_{5}(s)\frac{\theta^{2}(s)\left[M_{2}(s)+M_{5}(s)+M_{8}(s)\right]}{M_{2}(s)\sigma^{2}(s)}\right\} \ep_{t}\left[S^{*}(s)\right]\nonumber \\
+\left\{ m_{6}(s)+\left[r(s)+\alpha(s)\right]M_{6}(s)+\eta(s)M_{5}(s)\right.\nonumber \\
\left.-M_{5}(s)\frac{\theta^{2}(s)\left[M_{3}(s)+M_{6}(s)+M_{9}(s)\right]-\theta(s)\sigma(s)\rho(s)\beta(s)\left[M_{2}(s)-M_{3}(s)\right]}{M_{2}(s)\sigma^{2}(s)}\right\} \ep_{t}\left[L(s)\right]\nonumber \\
+\left\{ m_{7}(s)+r(s)M_{7}(s)\right\} L^{-\lambda}(t)+\left\{ m_{8}(s)+r(s)M_{8}(s)\right\} S^{*}(t)+\left\{ m_{9}(s)+r(s)M_{9}(s)\right\} L(t)\nonumber \\
+m_{10}(s)+r(s)M_{10}(s)-\frac{\theta^{2}(s)M_{10}(s)}{\sigma^{2}(s)}-M_{5}(s)\frac{\theta^{2}(s)M_{10}(s)}{M_{2}(s)\sigma^{2}(s)} & =0.\label{eq:2-4-8}
\end{align}
From (\ref{eq:2-4-8}), we can get the following equations for $M_{i},i=1,\cdots,10$:
\begin{equation}
\begin{cases}
m_{1}(s)+\left(r(s)-\lambda\left[\alpha(s)-\frac{1}{2}(\lambda+1)\beta^{2}(s)\right]\right)M_{1}(s)\\
\q-\frac{\theta^{2}(s)\left[M_{1}(s)+M_{4}(s)+M_{7}(s)\right]-\lambda\theta(s)\sigma(s)\rho(s)\beta(s)M_{1}(s)}{\sigma^{2}(s)} & =0,\q s\in[0,T],\\
\\
M_{1}(T)=0;
\end{cases}\label{eq:2-4-9}
\end{equation}
\begin{equation}
\begin{cases}
m_{2}(s)+2r(s)M_{2}(s)-\frac{\theta^{2}(s)\left[M_{2}(s)+M_{5}(s)+M_{8}(s)\right]}{\sigma^{2}(s)} & =0,\q s\in[0,T],\\
M_{2}(T)=1;
\end{cases}\label{eq:2-4-10}
\end{equation}
\begin{equation}
\begin{cases}
m_{3}(s)+\left[r(s)+\alpha(s)\right]M_{3}(s)+\eta(s)M_{2}(s)\\
\q-\frac{\theta^{2}(s)\left[M_{3}(s)+M_{6}(s)+M_{9}(s)\right]-\theta(s)\sigma(s)\rho(s)\beta(s)\left[M_{2}(s)-M_{3}(s)\right]}{\sigma^{2}(s)}=0, & s\in[0,T],\\
M_{3}(T)=0;
\end{cases}\label{eq:2-4-11}
\end{equation}
\begin{equation}
\begin{cases}
m_{4}(s)+\left(r(s)-\lambda\left[\alpha(s)-\frac{1}{2}(\lambda+1)\beta^{2}(s)\right]\right)M_{4}(s)\\
\q-M_{5}(s)\frac{\theta^{2}(s)\left[M_{1}(s)+M_{4}(s)+M_{7}(s)\right]-\lambda\theta(s)\sigma(s)\rho(s)\beta(s)M_{1}(s)}{M_{2}(s)\sigma^{2}(s)}=0, & s\in[0,T],\\
M_{4}(T)=0;
\end{cases}\label{eq:2-4-12}
\end{equation}
\begin{equation}
\begin{cases}
m_{5}(s)+2r(s)M_{5}(s)-M_{5}(s)\frac{\theta^{2}(s)\left[M_{2}(s)+M_{5}(s)+M_{8}(s)\right]}{M_{2}(s)\sigma^{2}(s)}=0, & s\in[0,T],\\
M_{5}(T)=-1;
\end{cases}\label{eq:2-4-13}
\end{equation}
\begin{equation}
\begin{cases}
m_{6}(s)+\left[r(s)+\alpha(s)\right]M_{6}(s)+\eta(s)M_{5}(s)\\
\q-M_{5}(s)\frac{\theta^{2}(s)\left[M_{3}(s)+M_{6}(s)+M_{9}(s)\right]-\theta(s)\sigma(s)\rho(s)\beta(s)\left[M_{2}(s)-M_{3}(s)\right]}{M_{2}(s)\sigma^{2}(s)}=0, & s\in[0,T],\\
M_{6}(T)=0;
\end{cases}\label{eq:2-4-14}
\end{equation}
\begin{equation}
\begin{cases}
m_{7}(s)+r(s)M_{7}(s)=0, & s\in[0,T],\\
M_{7}(T)=-\omega_{1};
\end{cases}\label{eq:2-4-15}
\end{equation}
\begin{equation}
\begin{cases}
m_{8}(s)+r(s)M_{8}(s)=0, & s\in[0,T],\\
M_{8}(T)=0;
\end{cases}\label{eq:2-4-16}
\end{equation}
\begin{equation}
\begin{cases}
m_{9}(s)+r(s)M_{9}(s)=0, & s\in[0,T],\\
M_{9}(T)=0;
\end{cases}\label{eq:2-4-17}
\end{equation}
\begin{equation}
\begin{cases}
m_{10}(s)+r(s)M_{10}(s)-\frac{\theta^{2}(s)M_{10}(s)}{\sigma^{2}(s)}-M_{5}(s)\frac{\theta^{2}(s)M_{10}(s)}{M_{2}(s)\sigma^{2}(s)}=0, & s\in[0,T],\\
M_{10}(T)=-\omega_{2}.
\end{cases}\label{eq:2-4-18}
\end{equation}

In the rest of this subsection, we focus on solving ODEs (\ref{eq:2-4-9})-(\ref{eq:2-4-18}).
First, from ODEs (\ref{eq:2-4-15})-(\ref{eq:2-4-17}), it is easy
to see that
\begin{align}
M_{7}(s) & =-\omega_{1}e^{\int_{s}^{T}r(y)\dt y},\q M_{8}(s)=M_{9}(s)\equiv0,\label{eq:2-4-19}
\end{align}
for $0\leq s\leq T$.

Second, it follows from (\ref{eq:2-4-10}) and (\ref{eq:2-4-13})
that $M_{2}(s)=-M_{5}(s)$, for $0\leq s\leq T$. Consequently, we
have 
\begin{equation}
M_{2}(s)=e^{\int_{s}^{T}2r(y)\dt y},\q M_{5}(s)=-e^{\int_{s}^{T}2r(y)\dt y}.\label{eq:2-4-20}
\end{equation}
Putting (\ref{eq:2-4-20}) into (\ref{eq:2-4-18}) yields that
\begin{equation}
M_{10}(s)=-\omega_{2}e^{\int_{s}^{T}r(y)\dt y}.\label{eq:2-4-21}
\end{equation}

With (\ref{eq:2-4-19}) and (\ref{eq:2-4-20}), we can get $\left(M_{1},M_{4}\right)$
and $\left(M_{3},M_{6}\right)$ from the systems of ODEs 
\begin{equation}
\begin{cases}
m_{1}(s)+\left(r(s)-\lambda\left[\alpha(s)-\frac{1}{2}(\lambda+1)\beta^{2}(s)\right]-\frac{\theta^{2}(s)-\lambda\theta(s)\sigma(s)\rho(s)\beta(s)}{\sigma^{2}(s)}\right)M_{1}(s)\\
\q-\frac{\theta^{2}(s)}{\sigma^{2}(s)}M_{4}(s)-\frac{\theta^{2}(s)}{\sigma^{2}(s)}M_{7}(s)=0, & s\in[0,T],\\
m_{4}(s)+\left(r(s)-\lambda\left[\alpha(s)-\frac{1}{2}(\lambda+1)\beta^{2}(s)\right]+\frac{\theta^{2}(s)}{\sigma^{2}(s)}\right)M_{4}(s)\\
\q+\frac{\theta^{2}(s)-\lambda\theta(s)\sigma(s)\rho(s)\beta(s)}{\sigma^{2}(s)}M_{1}(s)+\frac{\theta^{2}(s)}{\sigma^{2}(s)}M_{7}(s)=0, & s\in[0,T],\\
M_{1}(T)=0,\q M_{4}(T)=0
\end{cases}\label{eq:2-4-22}
\end{equation}
and 
\begin{equation}
\begin{cases}
m_{3}(s)+\left[r(s)+\alpha(s)-\frac{\theta^{2}(s)+\theta(s)\sigma(s)\rho(s)\beta(s)}{\sigma^{2}(s)}\right]M_{3}(s)-\frac{\theta^{2}(s)}{\sigma^{2}(s)}M_{6}(s)\\
\q+\left[\eta(s)+\frac{\theta(s)\rho(s)\beta(s)}{\sigma(s)}\right]M_{2}(s)=0, & s\in[0,T],\\
m_{6}(s)+\left[r(s)+\alpha(s)+\frac{\theta^{2}(s)}{\sigma^{2}(s)}\right]M_{6}(s)+\frac{\theta^{2}(s)+\theta(s)\sigma(s)\rho(s)\beta(s)}{\sigma^{2}(s)}M_{3}(s)\\
\q-\left[\eta(s)+\frac{\theta(s)\rho(s)\beta(s)}{\sigma(s)}\right]M_{2}(s)=0, & s\in[0,T],\\
M_{3}(T)=0,\q M_{6}(T)=0,
\end{cases}\label{eq:2-4-23}
\end{equation}
respectively. It follows from (\ref{eq:2-4-22}) and (\ref{eq:2-4-23})
that
\begin{eqnarray}
M_{1}(s) & = & -M_{4}(s)\nonumber \\
 & = & \exp\left\{ \int_{s}^{T}\left(r(y)-\lambda\left[\alpha(y)-\frac{1}{2}(\lambda+1)\beta^{2}(y)\right]+\frac{\lambda\theta(y)\rho(y)\beta(y)}{\sigma(y)}\right)\dt y\right\} \nonumber \\
 &  & \times\int_{s}^{T}\exp\left\{ -\int_{z}^{T}\left(r(y)-\lambda\left[\alpha(y)-\frac{1}{2}(\lambda+1)\beta^{2}(y)\right]+\frac{\lambda\theta(y)\rho(y)\beta(y)}{\sigma(y)}\right)\dt y\right\} \left[-\frac{\theta^{2}(z)}{\sigma^{2}(z)}M_{7}(z)\right]\dt z\nonumber \\
 & = & \omega_{1}e^{\int_{s}^{T}r(y)\dt y}\int_{s}^{T}\exp\left\{ \int_{z}^{s}\left(\lambda\left[\alpha(y)-\frac{1}{2}(\lambda+1)\beta^{2}(y)\right]-\frac{\lambda\theta(y)\rho(y)\beta(y)}{\sigma(y)}\right)\dt y\right\} \frac{\theta^{2}(z)}{\sigma^{2}(z)}\dt z,\label{eq:3-30}
\end{eqnarray}
and 
\begin{eqnarray}
M_{3}(s) & = & -M_{6}(s)\nonumber \\
 & = & \exp\left\{ \int_{s}^{T}\left[r(y)+\alpha(y)-\frac{\theta(y)\rho(y)\beta(y)}{\sigma(y)}\right]\dt y\right\} \nonumber \\
 &  & \times\int_{s}^{T}\exp\left\{ -\int_{z}^{T}\left[r(y)+\alpha(y)-\frac{\theta(y)\rho(y)\beta(y)}{\sigma(y)}\right]\dt y\right\} \left[\eta(z)+\frac{\theta(z)\rho(z)\beta(z)}{\sigma(z)}\right]M_{2}(z)\dt z\nonumber \\
 & = & e^{\int_{s}^{T}2r(y)\dt y}\int_{s}^{T}\exp\left\{ \int_{z}^{s}\left[\eta(y)+\frac{\theta(y)\rho(y)\beta(y)}{\sigma(y)}\right]\dt y\right\} \left[\eta(z)+\frac{\theta(z)\rho(z)\beta(z)}{\sigma(z)}\right]\dt z,\label{eq:3-31}
\end{eqnarray}
respectively.

\subsection{The equilibrium strategy}

From (\ref{eq:2-4-7}) and the results given by last subsection, we
have 
\begin{align*}
f_{1}(s) & =-\frac{\theta(s)M_{7}(s)-\lambda\sigma(s)\rho(s)\beta(s)M_{1}(s)}{\sigma^{2}(s)M_{2}(s)},\\
f_{2}(s) & =0,\\
f_{3}(s) & =\frac{\rho(s)\beta(s)}{\sigma(s)}\left[1-\frac{M_{3}(s)}{M_{2}(s)}\right],\\
f_{4}(s) & =-\frac{\theta(s)M_{10}(s)}{\sigma^{2}(s)M_{2}(s)}.
\end{align*}

\begin{thm}
\label{thm:3-1}Let
\[
M_{2}(s)=e^{\int_{s}^{T}2r(y)\dt y},\q M_{7}(s)=-\omega_{1}e^{\int_{s}^{T}r(y)\dt y},\q M_{10}(s)=-\omega_{2}e^{\int_{s}^{T}r(y)\dt y},
\]
$M_{1}$ and $M_{3}$ be given by (\ref{eq:3-30}) and (\ref{eq:3-31}),
respectively. Then the strategy defined by
\begin{align*}
u^{*}(s) & =f_{1}(s)L^{-\lambda}(s)+f_{3}(s)L(s)+f_{4}(s)
\end{align*}
is an equilibrium strategy.\end{thm}
\begin{proof}
Define $p_{1}$, $p_{2}$ and $\left(\mbd k_{1},\mbd k_{2}\right)$
by (\ref{eq:2-4-2}), (\ref{eq:2-4-3}) and (\ref{eq:2-4-6}), respectively.
Obviously, $\left(u^{*},\mbd X^{*},\mbd p,\left(\mbd k_{1},\mbd k_{2}\right)\right)$
satisfies the system (\ref{eq:2-3-4}). Furthermore, it is easy to
see that $f_{1}(s),\ f_{3}(s)$ and $f_{4}(s)$ are uniformly bounded.
Thus, we have $\mbd X^{*}\in L_{\mcl F}^{2}(\Omega;C(0,T,\mbb R^{2}))$
and $u^{*}\in L_{\mcl F}^{2}(0,T;\mbb R)$. 

Now, we are going to check whether the condition (\ref{eq:2-3-5})
is satisfied. Note that

\begin{eqnarray*}
\Lambda(s;t) & = & \theta(s)p_{1}(s;t)+\sigma(s)k_{1,1}(s;t)\\
 & = & \theta(s)\left\{ M_{1}(s)L^{-\lambda}(s)+M_{2}(s)S^{*}(s)+M_{3}(s)L(s)\right.\\
 &  & -M_{1}(s)\ep_{t}\left[L^{-\lambda}(s)\right]-M_{2}(s)\ep_{t}\left[S^{*}(s)\right]\\
 &  & \left.-M_{3}(s)\ep_{t}\left[L(s)\right]+M_{7}(s)L^{-\lambda}(t)+M_{10}(s)\right\} \\
 &  & +\sigma(s)\left\{ -\lambda\rho(s)\beta(s)M_{1}(s)L^{-\lambda}(s)+M_{2}(s)\left[\sigma(s)u^{*}(s)-\rho(s)\beta(s)L(s)\right]\right.\\
 &  & \left.+\rho(s)\beta(s)M_{3}(s)L(s)\right\} \\
 & = & \theta(s)M_{1}(s)\left[L^{-\lambda}(s)-\ep_{t}\left[L^{-\lambda}(s)\right]\right]+\theta(s)M_{2}(s)\left[S^{*}(s)-\ep_{t}\left[S^{*}(s)\right]\right]\\
 &  & +\theta(s)M_{3}(s)\left[L(s)-\ep_{t}\left[L(s)\right]\right]+\theta M_{7}(s)\left[L^{-\lambda}(t)-L^{-\lambda}(s)\right].
\end{eqnarray*}
Obviously, $\Lambda$ satisfies the first condition in (\ref{eq:2-3-5}).
It follows from
\[
\lim_{s\downarrow t}\ep_{t}\left[\left|L^{-\lambda}(s)-\ep_{t}\left[L^{-\lambda}(s)\right]\right|\right]=0,\ \mathrm{and}\ \lim_{s\downarrow t}\ep_{t}\left[\left|L^{-\lambda}(s)-L^{-\lambda}(t)\right|\right]=0,
\]
\[
\lim_{s\downarrow t}\ep_{t}\left[\left|\mbd X^{*}(s)-\ep_{t}\left[\mbd X^{*}(s)\right]\right|\right]=\mbd0,\q\text{and}\q\lim_{s\downarrow t}\ep_{t}\left[\left|\mbd X^{*}(s)-\mbd X^{*}(t)\right|\right]=\mbd0
\]
that $\Lambda$ satisfies the second condition in (\ref{eq:2-3-5}).\end{proof}
\begin{rem}
\label{rmk:3-2}Although we are looking for the equilibrium strategy
$u^{*}$ among the open-loop controls, it is a feedback control of
$L^{-\lambda}$ and $L$. Recall that the equilibrium strategy obtained
in \citet{wwyy12} is only a linear feedback control of the liability.
The results are different because the risk aversion considered in
this paper depends on $L^{-\lambda}$, while a constant risk aversion
is considered in \citet{wwyy12} if there is only one regime. 

If $\omega_{1}=0$, then
\begin{align*}
f_{1}(s) & =0,\\
f_{3}(s) & =\frac{\rho(s)\beta(s)}{\sigma(s)}\left[1-\frac{M_{3}(s)}{M_{2}(s)}\right],\\
f_{4}(s) & =-\frac{\theta(s)M_{10}(s)}{\sigma^{2}(s)M_{2}(s)},
\end{align*}
which means that equilibrium strategy $u^{*}$ always depends on the
liability, even if the risk aversion is independent of the liability.
Furthermore, it is interesting that we get the same equilibrium strategy
with \citet{wwyy12} in this special case (see Appendix).
\end{rem}

\section{The Equilibrium Value Function}

In this section, we are going to derive the equilibrium value function
$V$ which is defined by (\ref{eq:2-2-0-1}). The techniques are similar
to \citet{cl06}. To simplify the notation, we suppress the superscript
of $S^{*}$.

We can rewrite $S$ by 

\[
\begin{cases}
\dt S(s) & =\left\{ r(s)S(s)+\left[\eta(s)+\theta(s)f_{3}(s)\right]L(s)+\theta(s)f_{1}(s)L^{-\lambda}(s)+\theta(s)f_{4}(s)\right\} \dt s\\
 & \q+\left\{ \left[\sigma(s)f_{3}(s)-\rho(s)\beta(s)\right]L(s)+\sigma(s)f_{1}(s)L^{-\lambda}(s)+\sigma(s)f_{4}(s)\right\} \dt W_{1}(s)\\
 & \q-\sqrt{1-\rho^{2}(s)}\beta(s)L(s)\dt W_{2}(s),\\
S(0) & =s_{0}.
\end{cases}
\]

Recall that 
\[
\begin{cases}
\dt L(s) & =L(s)\left[\alpha(s)\dt s+\rho(s)\beta(s)\dt W_{1}(s)+\sqrt{1-\rho^{2}(s)}\beta(s)\dt W_{2}(s)\right],\\
L(0) & =l_{0}.
\end{cases}
\]
Let $\xi(s,q):=\theta(s)+q\sigma(s)\rho(s)\beta(s)$, for all $q\in\mathbb{R}$.
Then for all $q\in\mathbb{R}$ applying the Itô formula, we derive
the following SDEs for $L^{q},\ SL^{q},\ $ and $S^{2}$: 
\[
\begin{cases}
\dt L^{q}(s) & =qL^{q}(s)\left\{ \left[\alpha(s)-\frac{1}{2}(1-q)\beta^{2}(s)\right]\dt s+\left[\rho(s)\beta(s)\dt W_{1}(s)+\sqrt{1-\rho^{2}(s)}\beta(s)\dt W_{2}(s)\right]\right\} ,\\
L^{q}(0) & =l_{0}^{q};
\end{cases}
\]
\[
\begin{cases}
\dt S(s)L^{q}(s) & =\left\{ \left[r(s)+q\left[\alpha(s)-\frac{1}{2}(1-q)\beta^{2}(s)\right]\right]S(s)L^{q}(s)\right.\\
 & \q+\left[\eta(s)+\xi(s,q)f_{3}(s)-q\beta^{2}(s)\right]L^{q+1}(s)\\
 & \q+\xi(s,q)f_{1}(s)L^{-\lambda+q}(s)\\
 & \q\left.+\xi(s,q)f_{4}(s)L^{q}(s)\right\} \dt s\\
 & \q+(\cdots)\dt W_{1}(s)+(\cdots)\dt W_{2}(s),\\
S(0)L^{q}(0) & =s_{0}l_{0}^{q};
\end{cases}
\]
\[
\begin{cases}
\dt S^{2}(s) & =\left\{ 2r(s)S^{2}(s)+2\theta(s)f_{4}(s)S(s)\right.\\
 & \q+2\left\{ \eta(s)+\theta(s)f_{3}(s)\right\} S(s)L(s)\\
 & \q+\left[\sigma^{2}(s)f_{3}^{2}(s)-2\rho(s)\beta(s)\sigma(s)f_{3}(s)+\beta^{2}(s)\right]L^{2}(s)\\
 & \q+2\theta(s)f_{1}(s)S(s)L^{-\lambda}(s)+\sigma^{2}(s)f_{1}^{2}(s)L^{-2\lambda}(s)\\
 & \q+2\sigma(s)\left[\sigma(s)f_{3}(s)-\rho(s)\beta(s)\right]f_{1}(s)L^{-\lambda+1}(s)\\
 & \q+2\sigma^{2}(s)f_{1}(s)f_{4}(s)L^{-\lambda}(s)\\
 & \q+\left.2\sigma(s)\left[\sigma(s)f_{3}(s)-\rho(s)\beta(s)\right]f_{4}(s)L(s)+\left[\sigma(s)f_{4}(s)\right]{}^{2}\right\} \dt s\\
 & \q+(\cdots)\dt W_{1}(s)+(\cdots)\dt W_{2}(s),\\
S^{2}(0) & =s_{0}^{2}.
\end{cases}
\]
Therefore, for $s\in[t,T]$ and $q\in\mathbb{R}$ we obtain
\begin{equation}
\begin{cases}
\dt\ep_{t}\left[L^{q}(s)\right] & =q\left[\alpha(s)-\frac{1}{2}(1-q)\beta^{2}(s)\right]\ep_{t}\left[L^{q}(s)\right]\dt s,\\
\ep_{t}\left[L^{q}(t)\right] & =l_{t}^{q};
\end{cases}\label{eq:2-1}
\end{equation}
\begin{equation}
\begin{cases}
\dt\ep_{t}[S(s)] & =\left\{ r(s)\ep_{t}[S(s)]+\left[\eta(s)+\theta(s)f_{3}(s)\right]\ep_{t}[L(s)]\right.\\
 & \left.+\theta(s)f_{1}(s)\ep_{t}\left[L^{-\lambda}(s)\right]+\theta(s)f_{4}(s)\right\} \dt s,\\
\ep_{t}[S(t)] & =s_{t};
\end{cases}\label{eq:2-3}
\end{equation}
\begin{equation}
\begin{cases}
\dt\ep_{t}\left[S(s)L^{q}(s)\right] & =\left\{ \left[r(s)+q\left[\alpha(s)-\frac{1}{2}(1-q)\beta^{2}(s)\right]\right]\ep_{t}\left[S(s)L^{q}(s)\right]\right.\\
 & \q+\left[\eta(s)+\xi(s,q)f_{3}(s)-q\beta^{2}(s)\right]\ep_{t}\left[L^{q+1}(s)\right]\\
 & \q+\xi(s,q)f_{1}(s)\ep_{t}\left[L^{-\lambda+q}(s)\right]\\
 & \q\left.+\xi(s,q)f_{4}(s)\ep_{t}\left[L^{q}(s)\right]\right\} \dt s\\
 & \q+(\cdots)\dt W_{1}(s)+(\cdots)\dt W_{2}(s),\\
\ep_{t}\left[S(t)L^{q}(t)\right] & =s_{t}l_{t}^{q};
\end{cases}\label{eq:2-4}
\end{equation}
and
\begin{equation}
\begin{cases}
\dt\ep_{t}[S^{2}(s)] & =\left\{ 2r(s)\ep_{t}[S^{2}(s)]+2\theta(s)f_{4}(s)\ep_{t}[S(s)]\right.\\
 & \q+2\left[\eta(s)+\theta(s)f_{3}(s)\right]\ep_{t}[S(s)L(s)]\\
 & \q+\left[\sigma^{2}(s)f_{3}^{2}(s)-2\rho(s)\beta(s)\sigma(s)f_{3}(s)+\beta^{2}(s)\right]\ep_{t}[L^{2}(s)]\\
 & \q+2\theta(s)f_{1}(s)\ep_{t}\left[S(s)L^{-\lambda}(s)\right]+\sigma^{2}(s)f_{1}^{2}(s)\ep_{t}\left[L^{-2\lambda}(s)\right]\\
 & \q+2\sigma(s)\left[\sigma(s)f_{3}(s)-\rho(s)\beta(s)\right]f_{1}(s)\ep_{t}\left[L^{-\lambda+1}(s)\right]\\
 & \q+2\sigma^{2}(s)f_{1}(s)f_{4}(s)\ep_{t}\left[L^{-\lambda}(s)\right]\\
 & \q+\left.2\sigma(s)\left[\sigma(s)f_{3}(s)-\rho(s)\beta(s)\right]f_{4}(s)\ep_{t}[L(s)]+\left[\sigma(s)f_{4}(s)\right]{}^{2}\right\} \dt s,\\
\ep_{t}[S^{2}(t)] & =s_{t}^{2}.
\end{cases}\label{eq:2-5}
\end{equation}
Then solving ODEs (\ref{eq:2-1}), (\ref{eq:2-3}), and (\ref{eq:2-4}),
we obtain
\begin{align}
\ep_{t}\left[L^{q}(T)\right] & =l_{t}^{q}e^{\int_{t}^{T}q\left[\alpha(y)-\frac{1}{2}(1-q)\beta^{2}(y)\right]\dt y};\\
\ep_{t}\left[S(T)\right] & =s_{t}e^{\int_{t}^{T}r(y)\dt y}+\tilde{S}_{I}(t,T)\ep_{t}[L(T)]+\tilde{S}_{II}(t,T)\ep_{t}\left[L^{-\lambda}(T)\right]+\tilde{S}_{III}(t,T),\label{eq:2-6}
\end{align}
where
\begin{align*}
\tilde{S}_{I}(t,T) & =\int_{t}^{T}\left[\eta(v)+\theta(v)f_{3}(v)\right]e^{\int_{v}^{T}\eta(y)\dt y}\dt v,\\
\tilde{S}_{II}(t,T) & =\int_{t}^{T}\theta(v)f_{1}(v)e^{\int_{v}^{T}\left[r(y)+\lambda\alpha(y)-\frac{1}{2}\lambda(\lambda+1)\beta^{2}(y)\right]\dt y}\dt v,\\
\tilde{S}_{III}(t,T) & =\int_{t}^{T}\theta(v)f_{4}(v)e^{\int_{v}^{T}r(y)\dt y}\dt v;
\end{align*}
and 
\begin{align*}
\ep_{t}[S(T)L^{q}(T)] & =s_{t}l_{t}^{q}e^{\int_{t}^{T}\left\{ r(y)+q\left[\alpha(y)-\frac{1}{2}(1-q)\beta^{2}(y)\right]\right\} \dt y}+\widetilde{SL}_{I}(t,T,q)\ep_{t}\left[L^{q+1}(T)\right]\\
 & \q+\widetilde{SL}_{II}(t,T,q)\ep_{t}\left[L^{-\lambda+q}(T)\right]+\widetilde{SL}_{III}(t,T,q)\ep_{t}\left[L^{q}(T)\right],
\end{align*}
where
\begin{align*}
\widetilde{SL}_{I}(t,T,q) & =\int_{t}^{T}\left[\eta(v)+\xi(v,q)f_{3}(v)-q\beta^{2}(v)\right]e^{\int_{v}^{T}\left[\eta(y)-q\beta^{2}(y)\right]\dt y}\dt v,\\
\widetilde{SL}_{II}(t,T,q) & =\int_{t}^{T}\xi(v,q)f_{1}(v)e^{\int_{v}^{T}\left[r(y)+\lambda\alpha(y)-\frac{1}{2}\lambda\left(\lambda-2q+1\right)\beta^{2}(y)\right]\dt y}\dt v,\\
\widetilde{SL}_{III}(t,T,q) & =\int_{t}^{T}\xi(v,q)f_{4}(v)e^{\int_{v}^{T}r(y)\dt y}\dt v.
\end{align*}
Then (\ref{eq:2-5}) can be rewritten by
\begin{equation}
\begin{cases}
\dt\ep_{t}[S^{2}(s)] & =\left\{ 2r(s)\ep_{t}[S^{2}(s)]+F(t,s)\ep_{t}[L(s)]+G(t,s)\ep_{t}[L^{2}(s)]+H(t,s)\ep_{t}\left[L^{-\lambda}(s)\right]\right.\\
 & \q\left.+J(t,s)\ep_{t}\left[L^{-\lambda+1}(s)\right]+K(t,s)\ep_{t}\left[L^{-2\lambda}(s)\right]+M(t,s)\right\} \dt s,\\
\ep_{t}[S^{2}(t)] & =s_{t}^{2},
\end{cases}\label{eq:2-5-1-1}
\end{equation}
where
\begin{align*}
F(t,s) & =2\theta(s)f_{4}(s)\tilde{S}_{I}(t,s)+2\left[\eta(s)+\theta(s)f_{3}(s)\right]\widetilde{SL}_{III}(t,s,1)+2\sigma(s)\left[\sigma(s)f_{3}(s)-\rho(s)\beta(s)\right]f_{4}(s),\\
G(t,s) & =2\left[\eta(s)+\theta(s)f_{3}(s)\right]\widetilde{SL}_{I}(t,s,1)+\left[\sigma^{2}(s)f_{3}^{2}(s)-2\rho(s)\beta(s)\sigma(s)f_{3}(s)+\beta^{2}(s)\right],\\
H(t,s) & =2\theta(s)f_{4}(s)\tilde{S}_{II}(t,s)+2\theta(s)f_{1}(s)\widetilde{SL}_{III}(t,s,-\lambda)+2\sigma^{2}(s)f_{1}(s)f_{4}(s),\\
J(t,s) & =2\left[\eta(s)+\theta(s)f_{3}(s)\right]\widetilde{SL}_{II}(t,s,1)+2\theta(s)f_{1}(s)\widetilde{SL}_{I}(t,s,-\lambda)\\
 & \q+2\sigma(s)\left[\sigma(s)f_{3}(s)-\rho(s)\beta(s)\right]f_{1}(s),\\
K(t,s) & =2\theta(s)f_{1}(s)\widetilde{SL}_{II}(t,s,-\lambda)+\sigma^{2}(s)f_{1}^{2}(s),\\
M(t,s) & =2\theta(s)f_{4}(s)\left[s_{t}e^{\int_{t}^{s}r(y)\dt y}+\tilde{S}_{III}(t,s)\right]+2s_{t}l_{t}e^{\int_{t}^{s}[r(y)+\alpha(y)]\dt y}\left[\eta(s)+\theta(s)f_{3}(s)\right]\\
 & \q+2s_{t}l_{t}^{-\lambda}e^{\int_{t}^{s}\left\{ r(y)-\lambda\left[\alpha(y)-\frac{1}{2}(1+\lambda)\beta^{2}(y)\right]\right\} \dt y}\theta(s)f_{1}(s)+\left[\sigma(s)f_{4}(s)\right]{}^{2}.
\end{align*}
Thus, we obtain
\begin{align}
\ep_{t}[S^{2}(T)] & =s_{t}^{2}e^{2\int_{t}^{T}r(y)\dt y}+\widetilde{S^{2}}_{I}(t,T)\ep_{t}[L(T)]+\widetilde{S^{2}}_{II}(t,T)\ep_{t}[L^{2}(T)]\nonumber \\
 & \q+\widetilde{S^{2}}_{III}(t,T)\ep_{t}\left[L^{-\lambda}(T)\right]+\widetilde{S^{2}}_{IV}(t,T)\ep_{t}\left[L^{-\lambda+1}(T)\right]\nonumber \\
 & \q+\widetilde{S^{2}}_{V}(t,T)\ep_{t}\left[L^{-2\lambda}(T)\right]+\widetilde{S^{2}}_{VI}(t,T),\label{eq:2-7}
\end{align}
where
\begin{align*}
\widetilde{S^{2}}_{I}(t,T) & =\int_{t}^{T}e^{\int_{v}^{T}\left[2r(y)-\alpha(y)\right]\dt y}F(t,v)\dt v,\\
\widetilde{S^{2}}_{II}(t,T) & =\int_{t}^{T}e^{\int_{v}^{T}\left[2\eta(y)-\beta^{2}(y)\right]\dt y}G(t,v)\dt v,\\
\widetilde{S^{2}}_{III}(t,T) & =\int_{t}^{T}e^{\int_{v}^{T}\left\{ 2r(y)+\lambda\left[\alpha(y)-\frac{1}{2}(\lambda+1)\beta^{2}(y)\right]\right\} \dt y}H(t,v)\dt v,\\
\widetilde{S^{2}}_{IV}(t,T) & =\int_{t}^{T}e^{\int_{v}^{T}\left\{ 2r(y)+(\lambda-1)\left[\alpha(y)-\frac{1}{2}\lambda\beta^{2}(y)\right]\right\} \dt y}J(t,v)\dt v,\\
\widetilde{S^{2}}_{V}(t,T) & =\int_{t}^{T}e^{\int_{v}^{T}\left\{ 2r(y)+2\lambda\left[\alpha(y)-\frac{1}{2}(2\lambda+1)\beta^{2}(y)\right]\right\} \dt y}K(t,v)\dt v,\\
\widetilde{S^{2}}_{VI}(t,T) & =\int_{t}^{T}e^{\int_{v}^{T}2r(y)\dt y}M(t,v)\dt v.
\end{align*}
In summary, we have the following result.
\begin{prop}
Let $\left(S,L\right)$ be the solution to the SDE (\ref{eq:2-1-1})
with $u$ replaced by the equilibrium strategy $u^{*}$. The equilibrium
value function is given by 
\[
V(t,(s_{t},l_{t}))=\frac{1}{2}\ep_{t}\left[S^{2}(T)\right]-\frac{1}{2}\left(\ep_{t}\left[S(T)\right]\right)^{2}-\left(\omega_{1}l_{t}^{-\lambda}+\omega_{2}\right)\ep_{t}\left[S(T)\right],
\]
where $\ep_{t}\left[S(T)\right]$ and $\ep_{t}\left[S^{2}(T)\right]$
are given by (\ref{eq:2-6}) and (\ref{eq:2-7}), respectively.
\end{prop}

\section{Numerical Examples}

In this section, we illustrate our results by some numerical examples.
The comparisons between the equilibrium strategy and the pre-committed
strategy, and between the equilibrium value function and the pre-committed
optimal value function, are provided in \citet{wwyy12}. Recall that
we get the same result with \citet{wwyy12} in a special case. Thus
in this paper we do not make the comparison between our results and
the pre-committed strategy. We are concerned with the effect of the
state-dependent risk aversion on the equilibrium strategy.

All the parameters are listed blew:
\begin{align*}
T & =10,\q r=0.1,\q\mu=0.6,\q\sigma=0.3,\\
\alpha & =0.1,\q\beta=0.2,\q\rho=0.6.\q\lambda=0.5.
\end{align*}
In the following figures, three initial time points are chosen, i.e.,
$t=0,5,8$, and the surplus and the liability are 5 and 3, respectively. 

In Figure \ref{fig:5-1}, we plot the equilibrium strategy as well
as the equilibrium value function versus $\omega_{1}$ for different
$\omega_{2}$. It illustrates that the equilibrium strategy increases
as $\omega_{1}$ increases. This is reasonable, since the risk aversion
decreases as $\omega_{1}$ increases and the investor tends to invest
more into the stock market. The equilibrium value function is a decreasing
function of $\omega_{1}$. This implies that the investor can get
higher return by invests boldly (the risk aversion decreases as $\omega_{1}$
increases).

Figure \ref{fig:5-2} illustrates the equilibrium strategy and the
equilibrium value function versus $\omega_{2}$ for different $\omega_{1}$.
The curves of the equilibrium strategy and the equilibrium value function
show the same feature as in Figure \ref{fig:5-1}. 

\begin{figure}[b]
\subfloat[$t=0$]{\includegraphics[scale=0.8]{lam0\lyxdot 5-1-0}\includegraphics[scale=0.8]{V-lam0\lyxdot 5-1-0}

}

\subfloat[$t=5$]{\includegraphics[scale=0.8]{lam0\lyxdot 5-1-5}\includegraphics[scale=0.8]{V-lam0\lyxdot 5-1-5}

}

\subfloat[$t=8$]{\includegraphics[scale=0.8]{lam0\lyxdot 5-1-8}\includegraphics[scale=0.8]{V-lam0\lyxdot 5-1-8}}

\caption{\label{fig:5-1}Equilibrium strategy and equilibrium value function
versus $\omega_{1}$}
\end{figure}

\begin{figure}[b]
\subfloat[$t=0$]{\includegraphics[scale=0.8]{lam0\lyxdot 5-2-0}\includegraphics[scale=0.8]{V-lam0\lyxdot 5-2-0}}

\subfloat[$t=5$]{\includegraphics[scale=0.8]{lam0\lyxdot 5-2-5}\includegraphics[scale=0.8]{V-lam0\lyxdot 5-2-5}}

\subfloat[$t=8$]{\includegraphics[scale=0.8]{lam0\lyxdot 5-2-8}\includegraphics[scale=0.8]{V-lam0\lyxdot 5-2-8}}

\caption{\label{fig:5-2}Equilibrium strategy and equilibrium value function
versus $\omega_{2}$}
\end{figure}

\appendix

\section*{Appendix }

\setcounter{section}{1} \setcounter{equation}{0} 

Consider a special case of \citet{wwyy12} with one regime, one bound
and one risk asset. The equilibrium strategy is given by
\begin{eqnarray*}
\hat{u}(t,s,l) & = & \frac{\beta(t)\rho(t)}{\sigma(t)}\left[1-e^{-\int_{t}^{T}r(y)\dt y}b(t)\right]l+\frac{\theta(t)}{\gamma\sigma^{2}(t)}e^{-\int_{t}^{T}r(y)\dt y},
\end{eqnarray*}
where $b(t)$ satisfies the linear of ODE:
\begin{align*}
\begin{cases}
\dot{b}(t) & =-\left[\alpha(t)-\frac{\theta(t)\beta(t)\rho(t)}{\sigma(t)}\right]b(t)-\left[\eta(t)+\frac{\theta(t)\beta(t)\rho(t)}{\sigma(t)}\right]e^{\int_{t}^{T}r(y)\dt y},\\
b(T) & =0.
\end{cases}
\end{align*}
The solution to the above ODE is given by
\begin{eqnarray*}
b(t) & = & e^{\int_{t}^{T}\left[\alpha(y)-\frac{\theta(y)\beta(y)\rho(y)}{\sigma(y)}\right]\dt y}\int_{t}^{T}e^{-\int_{z}^{T}\left[\alpha(y)-\frac{\theta(y)\beta(y)\rho(y)}{\sigma(y)}\right]\dt y}\left[\eta(z)+\frac{\theta(z)\beta(z)\rho(z)}{\sigma(z)}\right]e^{\int_{z}^{T}r(y)\dt y}\dt z\\
 & = & e^{\int_{t}^{T}r(y)\dt y}\int_{t}^{T}\left[\eta(z)+\frac{\theta(z)\beta(z)\rho(z)}{\sigma(z)}\right]e^{\int_{z}^{t}\left[\eta(y)+\frac{\theta(y)\beta(y)\rho(y)}{\sigma(y)}\right]\dt y}\dt z.
\end{eqnarray*}
Now consider the special case of our model with $\omega_{1}=0$. Note
that the risk aversion in \citet{wwyy12} is $\gamma=\frac{1}{\omega_{2}}$.
Thus we have
\begin{align*}
f_{1}(t) & =0,\\
f_{3}(t) & =\frac{\rho(t)\beta(t)}{\sigma(t)}\left[1-\int_{t}^{T}\left[\eta(z)+\frac{\theta(z)\rho(z)\beta(z)}{\sigma(z)}\right]e^{\int_{z}^{t}\left[\eta(y)+\frac{\theta(y)\rho(y)\beta(y)}{\sigma(y)}\right]\dt y}\dt z\right]\\
 & =\frac{\rho(t)\beta(t)}{\sigma(t)}\left[1-e^{-\int_{t}^{T}r(y)\dt y}b(t)\right],\\
f_{4}(t) & =\frac{\theta(t)}{\gamma\sigma^{2}(t)}e^{-\int_{t}^{T}r(y)\dt y}.
\end{align*}
Thus, in this special case, we get the same equilibrium strategy.

\section*{Acknowledgments}

We would like to thank the referee for valuable comments and suggestions.
This work was supported by National Natural Science Foundation of
China (10971068), Doctoral Program Foundation of the Ministry of Education
of China (20110076110004), Program for New Century Excellent Talents
in University (NCET-09-0356) and the Fundamental Research Funds for
the Central Universities.

\bibliographystyle{plainnat}
\bibliography{ZWW-1206}

\end{document}